\documentclass[11pt, reco]{article}
\usepackage{amssymb,amscd,amsmath,amsthm}
\usepackage{graphicx}
\graphicspath{ {./images/} }
\newtheorem{theorem}{Theorem}[section]

\newtheorem{lemma}[theorem]{Lemma}
\newtheorem{definition}[theorem]{Definition}

\newtheorem{remark}[theorem]{Remark}

\def\ce{{\mathcal E}}

\def\bc{{\mathbb C}}

\def\bn{{\mathbb N}}

\def\a{\alpha}
\def\b{\beta}

\def\tr{{\rm Tr}}
\def\L{\Lambda}
\def\G{\Gamma}
\def\ce{\mathcal E}

\def\ffi{\varphi}

\def\Tr{\mathrm{Tr}}
\def\<{\langle}
\def\>{\rangle}

\def\1{\mathbf{1}}

\def\cal{\mathcal}

\def\id{{\bf 1}\!\!{\rm I}}
\begin{document}

\title{Refinement of quantum Markov states on trees }

\begin{center}
{\Large {\bf Refinement of quantum Markov states on trees}}\\[1cm]
\end{center}

\begin{center}
{\large {\sc Farrukh Mukhamedov$^a$}}\\[2mm]
\textit{
$^a$Department of Mathematical Sciences,\\
College of Science, United Arab Emirates University,  \\
P.O. Box 15551,Al Ain, Abu Dhabi, UAE}\\
E-mail: {\tt far75m@yandex.ru, \ farrukh.m@uaeu.ac.ae}
\end{center}

\begin{center}
{\sc Abdessatar Souissi$^{bc}$}\\
\textit{
$^b$ Department of Accounting, College of Business Management\\
Qassim University, Ar Rass, Saudi Arabia \\
$^c$ Preparatory institute for scientific and technical studies,\\
 Carthage University, La Marsa, Tunisia}\\
E-mail: {\tt a.souaissi@qu.edu.sa}\\

\end{center}

\begin{abstract}
In the present paper, we propose a refinement for the notion of quantum Markov states (QMS) on trees. A structure theorem for QMS on general trees is proved. We notice that any restriction of QMS in the sense of Ref. \cite{AccFid03} is not necessarily to be a QMS. It turns out that localized QMS has the mentioned property which is called \textit{sub-Markov states}, this allows us to characterize translation invariant QMS on regular trees.
\end{abstract}

\section{Introduction}\label{sec_intro}

The study of  quantum many-body systems has lived an explosion of results. This is specifically true in the field of Tensor
Networks. Recent studies show that  "Matrix Product States" and more generally  "Tensor Network States" play a crucial role in the description of the whole quantum system under consideration \cite{CV,Or}. This approach is based on the density matrix renormalization group (DMRG) algorithm.

On the other hand a physically interesting mathematical approach to quantum states on tensor networks was proposed by Accardi through introducing quantum Markov chains on tensor product of matrix algebras \cite{[Ac74f],Ac75,GZ}.  Since then quantum Markov chains have found a  great progress  and a number of applications in several research domains: computational physics where they are called \textit{Bethe ansatz states} \cite{[RoOs96]},
   spin models in quantum statistical physics \cite{F.N.W},  \cite{[Mohari18]} where they are called \textit{finitely correlated states},
interacting particle systems \cite{[AcKo00b-QIP]}, quantum information \cite{[CiPe-Ga-Sc-Ver17]}
where they are called \textit{matrix product states}, quantum random walks \cite{[AcWat87]},
\cite{WhiRoRoAspGu10}, cognitive sciences \cite{[AcKhreOhy09]}.

On the other hand, in \cite{MBS161,MBS162,MR2004,MR2005} a particular class of
quantum Markov chains (QMC) associated to the Ising types models on the Cayley
trees have been explored (see \cite{BR19,GMM09,RR17,Roz} for recent development on models over such trees). It turned out that the
above considered QMCs fall to a special class called
\textit{quantum Markov states (QMS)} (see \cite{MS19,MS20}). Furthermore, in \cite{FM,MS19,MS20} a description of
QMS has been carried out. It is worth to indicate that introduced
QMS were considered over the Cayley trees, and investigated the
Markov property not only with respect to levels of the considered tree, but
also with regard to the interaction domain at each site, which has a finer structure, and through a family of suitable
quasi-conditional expectations so-called \textit{localized}
\cite{MS19}. Such a localization property is essential for the integral decomposition of QMS, since takes into
account finer structure of conditional expectations and
filtration. If one considers conditional expectations without localization
property, then the results of \cite{AccFid03} can be applied to the considered QMS and one can get the
disintegration of QMS, which would be not enough for its finer representation.
 In the present paper, we are going to investigate such states from conceptual point of view, i.e., we aim
to define quantum Markov states (QMS) on the finer structure of a tree graph as a refinement of our previous works \cite{MS19}, \cite{MS20}.
It is stressed that the considered quantum Markov states do not have one-dimensional analogues, hence results of \cite{AccFid03} are not applicable. We notice that types of von Neumann algebras generated by QMS have been investigated in \cite{FM,M04,MBS162,MR2004,MS201}.

We point out that the present work is another step towards one of the most important open problems in quantum probability, which concerns the construction of a satisfactory theory of quantum Markov fields. This problem relates to an extension of the Dobrushin Markov fields \cite{D}
to the quantum setting. First attempts regarding this goal have been done in \cite{AOM}, \cite{[AMSo]}.
In this direction, quantum Markov chains on trees and
  their applications to quantum phase transition phenomena for concrete models were explored in an increasing number of works
  (see for instance \cite{AccMuSa1,AccMuSa2,AccMuSa3}, \cite{MBS161,MBS162,MBSS,MGS17,MGS19}).

In the present paper, we provide a conceptual new definition of QMS on trees that refines the definition introduced in \cite{MS19}. This allows to  introduce a notion of translation invariance for QMS on regular trees, known as Cayley trees (also Bethe lattice) \cite{Ost}.
We prove a structure theorem for QMS on trees extending a result of \cite{[AcSouElG20]}. We notice that any restriction of QMS in the sense of Definition \ref{QMS1}, is not necessarily to be QMS. It turns out that localized QMS has the mentioned property which is called \textit{sub-Markov states}, which allows us to characterize translation invariant QMS on regular trees.

It is stressed that the present work opens a new perspective for the generalization of  many interesting results related to one dimensional quantum Markov states and chains to multi-dimensional cases. Namely, the entropy of QMC \cite{AOW,Fid01,OW19,Park,PO} was established for translation invariant quantum Markov states. An other interesting problem concerns the open quantum random walks on trees as generalization of QMC \cite{DKY19,DM19}, \cite{Attal12}, \cite{DM18}.

Let us mention the outlines of this paper. After preliminaries notions on trees in section \ref{sec_prel}. Section \ref{sec_QMCS_tree} is devoted to  quantum Markov chains and states on trees. In section \ref{sec_stru_QMS} we prove a structure theorem for QMS on general tree graphs. Section \ref{sec_SubMS} is devoted to the second main result of the paper which concerns a characterization of translation invariant quantum Markov states on Cayley trees.

\section{Preliminaries}\label{sec_prel}

Let $T = (V,E)$ be a locally finite tree. We fix a root $o\in  V$.  Two vertices $x$ and $y$ are  {\it nearest neighbors} (denoted $x\sim y$ ) if they are joined through an edge (i.e. $<x,y>\in E$). A list $x\sim x_1\sim \dots \sim x_{d-1}\sim y$ of vertices is called a {\it
path} from $x$ to $y$. The distance on the tree $d(x,y)$  is the length of the shortest path from $x$ to $y$.
The set of  its \textit{direct successors}  of a given vertex $x\in V$  is defined  by
\begin{equation}\label{S(x)def}
S(x) :  = \left\{y\in V \, \,  : \, \,  x\sim y \, \, \hbox{and} \, \, d(y,o) > d(x,o) \right\}
\end{equation}
and its $k^{th}$ successors w.r.t. the root $o$  is defined by induction as follows
$$
S_1(x) := S(x);
$$
$$
  S_{k+1}(x) = S(S_k(x)),\, \,   \forall k\ge 1.
$$
The "future" w.r.t. the vertex $x$ is defined by:
\begin{equation}\label{S_1n_S_infty}
  S_{[m,n]}(x)= \bigcup_{k=m}^{n}S_{k}(x); \quad    T(x) = \bigcup_{k\ge1}S_{k}(x); \quad T^{'}(x) = T(x)\setminus \{x\}.
\end{equation}
 In the  homogeneous case ($|S(x)|=k$ is constant w.r.t. the vertex $x$), the graph $T$ coincides with  the semi-infinite Cayley tree $\Gamma^k_+$ of order $k$. In particular, if $k=1$  the graph is reduced to the one-side integer lattice $\mathbb N$.

Let $x\in V$. If $o=x_0\sim x_1 \sim \cdots x_n =x$ is the unique edge-path with minimal length joining $o$ and $x$,  the set
\begin{equation}\label{P(x)}
P(x):= \{x_0, x_1, \cdots, x_{n-1} \}
\end{equation}
 represents the set of "predecessors" of the vertex $x$ w.r.t. the root $o$. The author is referred to \cite{AS20} for a more detailed description of the hierarchical structure of rooted trees.

 Define
  \[\Lambda_n = S_n(o); \quad \Lambda_{[n,m]} := S_{[n,m]}(0);\quad \Lambda_{[0,n]} := S_{[0,n]}(0).\]
To each vertex $x$, we associate a $C^*$--algebra $\mathcal{A}_x$ with identity $\id_x$. For a given bounded region $\Lambda$, we consider the algebra $\mathcal{A}_{\Lambda} = \bigotimes_{x\in \Lambda}\mathcal{A}_x$. One can consider the following embedding
$$
\mathcal{A}_{\Lambda_{[0,n]}}\equiv  \mathcal{A}_{\Lambda_{[0,n]}}\otimes\id_{\Lambda_{n+1}}\subset \mathcal{A}_{\Lambda_{[0,n+1]}}.
$$
The algebra $\mathcal{A}_{\Lambda_{[0,n]}}$ can be viewed as a  subalgebra of $\mathcal{A}_{\Lambda_{[0,n+1]}}$. It follows the following quasi-local algebra.
\begin{equation}\label{AVloc}
  \mathcal{A}_{V;\, loc} := \bigcup_{n\in\mathbb{N}}\mathcal{A}_{[0,n]}
\end{equation}
and the quasi-local algebra
$$
\mathcal{A}_V := \overline{\mathcal{A}_{V;\, loc}}^{C^*}.
$$
The set of states on a $C^*$--algebra $\mathcal{A}$ will be denoted by $\mathcal{S}(\mathcal{A})$.

\section{ Quantum Markov chains and States on trees}\label{sec_QMCS_tree}

Let us consider a triplet ${\cal C} \subset {\cal B} \subset {\cal A}$ of
unital $C^*$-algebras. Recall \cite{[Ac74f]} that a {\it
quasi-conditional expectation} with respect to the given triplet
is a completely positive (CP) linear map $\ce \,:\, {\cal A} \to
{\cal B}$ such that $ \ce(ca) = c \ce(a)$, for all $a\in {\cal A},\, c \in {\cal C}$.

\begin{definition}\cite{AccFidMu,AOM}\label{QMCdef}
Let $\varphi$ be a state on $\mathcal{A}_V$. Then $\varphi$ is called a
{\it (backward) quantum Markov chain}, associated with $\{\L_n\}$, if
there exist a quasi-conditional expectation $E_{\Lambda_{[0,n]}}$ with
respect to the triple $\mathcal{A}_{{\Lambda}_{n-1]}}\subseteq \mathcal{A}_{\Lambda_{[0,n]}}\subseteq\mathcal{A}_{\Lambda_{[0,n+1]}}$ for each
$n\in\bn$ and an initial state $\varphi_0\in S(\mathcal{A}_{\L_0})$ such that
\begin{equation}\label{lim_Mc}
\varphi = \lim_{n\to\infty} \varphi_0\circ E_{\Lambda_{0]}}\circ
E_{\Lambda_{1]}} \circ \cdots \circ E_{\Lambda_{[0,n]}}
\end{equation}
in the weak-* topology.
\end{definition}

In \cite{AccFidMu} it was given a general definition of quantum Markov states which can  be adopted to the considered setting as follows.
 \begin{definition}\label{QMS1}
A quantum Markov chain $\ffi $ is said to be quantum Markov state
with respect to the sequence $\{\ce_{\Lambda_{j]}}\}$ of
quasi-conditional expectations if one has
\begin{equation}\label{1Markov_state_eq}
\ffi_{\lceil \mathcal{A}_{\Lambda_{j]}}}\circ \ce_{\Lambda_{j]}}=\ffi_{\lceil \mathcal{A}_{\Lambda_{j+1]}}}, \ \ \forall j\in\bn.
\end{equation}
\end{definition}

Using this definition, in \cite{AccFidMu} non-homogeneous QMS has been characterized. To formulate that result let us recall some notations.

Let us assume that we have a locally faithful state $\ffi$ on the quasi-local algebra $\mathcal{A}_V$. Then a
potential $h_{\L_n}$ is canonically defined for each finite subset $\L_n$  as follows
\begin{equation}\label{HL}
\ffi_{\lceil \mathcal{A}_{\Lambda_{[0,n]}}}=\tr_{\mathcal{A}_{\Lambda_{[0,n]}}}(e^{-h_{\L_n}}\cdot).
\end{equation}
Such a set of potentials $h_{\L_n}$ satisfies normalization conditions
$$
\tr_{\mathcal{A}_{\Lambda_{[0,n]}}}(e^{-h_{\L_n}})=1.
$$

Next result has been formulated and proved for trees in \cite{MS19}.

 \begin{theorem}\label{AFt}\cite{AccFidMu}
 Let $\ffi$ be a locally faithful state on  $\mathcal{A}_V$. Then the following statements are
 equivalent:
 \begin{enumerate}
 \item[(i)] $\ffi$ is a QMS w.r.t. the sequence  $\{\ce_{\Lambda_{j]}}\}$
 of transition expectations;
\item[(ii)] The sequence of potentials $\{h_{\L_n}\}$ associated to $\ffi$ by \eqref{HL}, can be recovered by
\begin{equation}\label{potentialequation}
 h_{\L_n}= H_{W_0} + \sum_{j=0}^{n-1}H_{W_j,W_{j+1}} + \hat H_{W_n}
 \end{equation}
where the sequences $\{H_{W_{j}}\}_{j\geq
0},\;\{\hat H_{W_{j}}\}_{j\geq 0}$
and $\{H_{W_{j}, W_{j+1}}\}_{j\geq 0}$ of self-adjoint operators localized in $\mathcal{A}_{W_j}$ and
$\mathcal{A}_{\L_{j,j+1}}$,
 respectively, and satisfying commutation relations
\begin{eqnarray}\label{commutation}
&&[H_{W_n}, H_{W_n,W_{n+1}}]=0,\quad  [ H_{W_n,W_{n+1}},
\hat H_{W_{n+1}}]=0,\nonumber\\ &&[H_{W_n}, \widehat{H}_{W_n}]=0,\quad
[H_{W_n,W_{n+1}}, H_{W_{n+1},W_{n+2}}]=0.
 \end{eqnarray}
 \end{enumerate}
 \end{theorem}

We stress that one considers QMS with respect to the levels of the tree, then QMS can be associated with potential
given in Theorem \ref{AFt}. However, the definition \ref{QMS1} does not take into account a finer structure of the tree.
This can be seen in the decomposition \eqref{potentialequation} since the terms $H_{W_j,W_{j+1}}$ are not specified.
Recently, in \cite{MS19} we have considered  QMS with finer structure of localized conditional expectations defined on the Cayley trees, then
the localization property \eqref{localiz} allowed us to explicitly find forms of $H_{W_j,W_{j+1}}$ and $\hat H_{W_j}$ in terms of nearest neighbor and competing interactions, namely
\begin{eqnarray}\label{HW1}
  &&H_{W_j,W_{j+1}}=\sum_{x\in W_j}H_{x,\overrightarrow{S(x)}},   \ \ \ \widehat{H}_{W_j}:=\sum_{x\in W_{j}}\widehat{H}_x.
 \end{eqnarray}
We again emphasize that if one considers conditional expectations without localization
property, then the expression is impossible to be obtained. Therefore,  it is natural to provide a conceptual definition QMS
which could cover such finer structure.

 \begin{definition}\label{QMS1}
Let $\ffi $ be a quantum  Markov chain on $\mathcal{A}_V$  w.r.t. a sequence $\{E_{\Lambda_{[0,n]}}\}_n$ of quasi-conditional expectations.
If for each $x\in V$ the restriction
\begin{equation}
E_x := E_{\Lambda_{[0,n]}}\lceil_{\mathcal{A}_{P(x)\cup \{x\}\cup S(x)}}
\end{equation}
defines a quasi-conditional   expectations w.r.t. the triplet $\mathcal{A}_{P(x)}\subset \mathcal{A}_{P(x)\cup \{x\}}\subset \mathcal{A}_{P(x)\cup\{x\}\cup \overrightarrow{S}(x)}$ such that
\begin{equation}\label{Markov_state_eq}
\ffi_{\lceil \mathcal{A}_{\{x\}\cup \overrightarrow{S}(x)}}\circ E_x =\ffi_{\lceil \mathcal{A}_{P(x)\cup\{x\}\cup \overrightarrow{S}(x)}}
\end{equation}
then $\varphi$ is called \textit{(localized) quantum Markov state  w.r.t. the family  $\{E_x \}_{x\in L}$} of local quasi-conditional expectations.
\end{definition}

\begin{remark}
The above definition refines the notion of quantum Markov states on trees introduced in \cite{MS19} for which the Markov property (\ref{Markov_state_eq}) is replaced by \eqref{1Markov_state_eq}.
In the one-dimensional case these two definitions coincide with the usual notion quantum Markov states introduced in \cite{[AcFr80]}.
\end{remark}

\subsection{Why localized quantum Markov states on trees are not reducible to the one dimensional case}

In this section, we are going to demonstrate that localized QMS cannot be reduced to the one dimensional case. More precisely,
we explicitly consider the potential \eqref{HW1}.

In what follows, we consider a semi-infinite
Cayley tree $\G^2_+=(L,E)$ of order two. Our starting
$C^{*}$-algebra is the same $\mathcal{A}_V$ but with $\mathcal{A}_{x}=M_{2}(\bc)$
for all $x\in V$.
By $"\Tr"$ we denote the normalized on the local algebra that assigns $1$  to the identity $\id$. For a given bounded region $\Lambda$, the associated partial trace is defined by
$$
\Tr_{\Lambda]}(a_\Lambda\otimes a_{\Lambda^c}) = a_{\Lambda}\Tr(a_{\Lambda^c}).
$$
 Denote
\begin{equation}\label{pauli} \id^{(u)}=\left(
          \begin{array}{cc}
            1 & 0 \\
            0 & 1 \\
          \end{array}
        \right), \quad
        \         \
\sigma^{(u)}= \left(
          \begin{array}{cc}
            1 & 0 \\
            0 & -1 \\
          \end{array}
        \right).
\end{equation}

For every vertices  $(x,(x,1),(x,2))$  we put
\begin{eqnarray}\label{1Kxy1}
&&K_{<x,(x,i)>}=\exp\{\b H_{x,(x,i)>}\}, \ \ i=1,2,\ \b>0,\\[2mm] \label{1Lxy1} &&
L_{>(x,1),(x,2)<}=\exp\{J\beta H_{>(x,1),(x,2)<}\}, \ \ J>0,
\end{eqnarray}
where
\begin{eqnarray}\label{1Hxy1}
&&
H_{<x,(x,i)>}=\frac{1}{2}\big(\id^{x)}\id^{(x,i)}+\sigma^{(x)}\sigma^{(x,i)}\big),\\[2mm]
\label{1H>xy<1} &&
H_{>(x,1),(x,2)<}=\frac{1}{2}\big(\id^{(x,1)}\id^{(x,2)}+\sigma^{(x,1)}\sigma^{(x,2)}\big).
\end{eqnarray}

The defined model is called  the {\it Ising model with competing
interactions} per vertices  $(u,(u,1),(u,2))$.

Therefore, one finds
\begin{eqnarray}\label{K<u,v>K_0K_3}
&&K_{<u,v>}=K_0\id^{(u)}\id^{(v)}+K_3\sigma^{(u)}\sigma^{(v)},\\[2mm]
&&\label{L>u,v<R_0R_3}
L_{>u,v<}=R_0\id^{(u)}\id^{(v)}+R_3\sigma^{(u)}\sigma^{(v)},
\end{eqnarray}
where
\begin{eqnarray*}
&&K_0=\frac{\exp{\beta}+1}{2},\ \ \   K_3=\frac{\exp{\beta}-1}{2},\\[2mm]
&&R_0=\frac{\exp{(J\beta)}+1}{2}, \ \ \
R_3=\frac{\exp{(J\beta)}-1}{2}.
\end{eqnarray*}
Let
\begin{equation}
A_{(x,(x,1), (x,2))} := K_{<x,(x,1)>}K_{<x,(x,2)>}L_{>(x,1), (x,2)<}
\end{equation}
The operator $A_{(x,(x,1), (x,2))}$ is localized on the algebra of observable associated with the ternary $(x, (x,1), (x,2))$.\\
A simple calculation leads to
\begin{eqnarray}\label{Ax}
A_{(x,(x,1),(x,2))}&=&\gamma\id^{(x)}\otimes\id^{(x,1)}\otimes\id^{(x,2)}+\delta\sigma^{(x)}\otimes\sigma^{(x,1)}\otimes\id^{(x,2)}\nonumber\\[2mm]
&&+
\delta\sigma^{(x)}\otimes\id^{(x,1)}\otimes\sigma^{(x,2)}+\eta\id^{(x)}\otimes\sigma^{(x,1)}\otimes\sigma^{(x,2)},
\end{eqnarray}
where
\begin{equation}
\left\{
  \begin{array}{ll}
    \gamma=K_{0}^{2}R_{0}+K_{3}^{2}.R_{3}=\frac{1}{4}[\exp{(J+2)\beta}+\exp{J\beta}+2\exp{\beta}], \\
    \\
    \delta=K_{0}K_{3}(R_{0}+R_{3})=\frac{1}{4}\exp{J\beta}[\exp{2\beta}-1],  \\
    \\
    \eta=K_{0}^{2}R_{3}+K_{3}^{2}R_{0}=\frac{1}{4}[\exp{(J+2)\beta}+\exp{J\beta}-2\exp{\beta}]. \\
  \end{array}
\right.
\end{equation}

In \cite{MBS161}, we have found that
\begin{equation}\label{dysys}
  \Tr_{u]}\left(A_{(u, (u,1), (u,2))}^{*}\id^{(u)}\otimes h^{(u,1)}\otimes h^{(u,2)}A_{(u, (u,1), (u,2))} \right) = h^{(u)}.
\end{equation}
where $h^{(x)}=h_\a$,
\begin{equation}
h_{\alpha}=
 \left(
  \begin{array}{cc}
    \alpha & 0 \\
    0 & \alpha \\
  \end{array}
\right)
\end{equation}
here $\alpha = \frac{4}{\exp(2J\beta)(\exp(4\beta)+1)+2\exp(2\beta)}.$

Define
\begin{equation}\label{Ealpha}
  \mathcal{E}^{\alpha}_{(u, (u,1), (u,2))}(a) = \Tr_{u]}\left(\alpha^{1/2}A_{(u, (u,1), (u,2))}^{*} a A_{(u, (u,1), (u,2))} \alpha^{1/2}\right).
  \end{equation}
  One can see that the map $\mathcal{E}^{\alpha}$ is an identity preserving transition expectation from $\mathcal{A}_{\{u\}\cup S(u)}$ into
  $\mathcal{A}_{u}$. Then  the map
\begin{equation}\label{Ealpha_x}
  E^{\alpha}_{u} = id_{P(u)}\otimes\mathcal{E}^{\alpha}_{(u, (u,1), (u,2))}
\end{equation}
is a quasi-conditional expectation w.r.t. the triplet $\mathcal{A}_{P(x)}\subset\mathcal{A}_{P(x)\cup \{x\}}\subset\mathcal{A}_{P(x)\cup\{x\}\cup
S(x)}$.

 Let  $\varphi_\alpha$ be the quantum Markov chain associated with $\{  E^{\alpha}_{u}\}$. Then
 \begin{equation}\label{phialpha}
\varphi_{\alpha}(a)=\alpha^{2^{n}-1}\Tr\bigg(a\prod_{i=0}^{n-1}K_{[i,i+1]}K_{[i,i+1]}^{*}\bigg),
\ \ \forall a\in \mathcal{A}_{\Lambda_{[0,n]}}.
\end{equation}
where $K_{[i,i+1]} = \prod_{u\in \Lambda_i}A_{(u, (u,1), (u,2))}$.\\

 \begin{theorem}
   The state $\varphi_\alpha$ is a quantum Markov state in the sense of Definition \ref{QMCdef} associated with the  quasi-conditional  expectations
   (\ref{Ealpha_x}).
 \end{theorem}

\begin{proof}
Since $\varphi_\alpha$ is a quantum Markov chain, it is enough to  show that it satisfies (\ref{Markov_state_eq}). Let $x\in \Lambda_n$ and $o=x_0\sim
x_1\sim\cdots x_n = x$ be the unique simple edge path joining the root $o$ to $x$. One can see that $x_{k+1} \in S(x_k)$ for every $k\in\{0,\cdots,
n-1\}$. Then there exists $i_k\in \{1,2\}$ such that $x_{k+1} = (x_k, i_k)$. For $a = a_0\otimes a_{x_1}\otimes \cdots\otimes a_{x_{n-1}}\otimes a_{x}\otimes
a_{(x,1)}\otimes a_{(x,2)}\in \mathcal{A}_{P(x)\cup\{x\}\cup S(x)}$ one has
\begin{eqnarray*}
  {\varphi_{\alpha}}_{\lceil\mathcal{A}_{P(x)\cup\{x\}\cup S(x)}}(a)  &=& \varphi_\alpha\left(a\otimes \id_{\Lambda_{[0,n+1]}\setminus\{P(x)\cup\{x\}\cup
  S(x)\}}\right)  \\
  &=& \Tr\bigg(\prod_{i=0}^{n}\Tr_{\Lambda_i]}\left(K_{[i,i+1]}a K_{[i,i+1]}^{*}\right)\bigg)\\
    &=& \alpha^{2^{n+1}-1}\Tr\left(\prod_{i=0}^{n}\prod_{x\in \Lambda_i} A_{(u, (u,1), (u,2))} a  \prod_{i=0}^{n}\prod_{x\in \Lambda_i}
    A_{(u, (u,1), (u,2))}^{*}\right)  \\
   &=&  \alpha^{2^{n+1}-1}\Tr\left(\prod_{u\in \Lambda_{[0,n]} }A_{(u, (u,1), (u,2))} a \prod_{u\in \Lambda_{[0,n]}} A_{(u, (u,1), (u,2))}^{*}\right)
\end{eqnarray*}
$$
\Tr_{\Lambda_n]}\left(K_{[n,n+1]} a K_{[n,n+1]}^{*}\right)  =  a_0\otimes a_{x_1}\otimes\cdots\otimes a_{x_{n-1}}\otimes \Tr_{\Lambda_{n}]}\left(K_{[n,n+1]} (a_{x}\otimes
a_{(x,1)}\otimes a_{(x,2)}\otimes\id) K_{[n,n+1]}^{*}\right)$$
$$
= a_0\otimes a_{x_1}\otimes\cdots\otimes a_{x_{n-1}}\otimes\bigotimes_{u\in \Lambda_n}\Tr_{\{u\}}\left(\prod_{u\in \Lambda_n}A_{(u, (u,1), (u,2))} a_{x}\otimes a_{(x,1)}\otimes
a_{(x,2)}\otimes \id A_{(u, (u,1), (u,2))}^{*}\right)
$$
$$
= a_0\otimes a_{x_1}\otimes\cdots\otimes a_{x_{n-1}}\otimes\Tr_{\{x\}}\left(A_{(x, (x,1), (x,2))} a_{x}\otimes a_{(x,1)}\otimes a_{(x,2)} A_{(x, (x,1),
(x,2))}^{*}\right)
$$
$$
\bigotimes_{u\in \Lambda_n\setminus\{x\}}\Tr_{\{u\}}\left(\prod_{u\in \Lambda_n}A_{(u, (u,1), (u,2))} \id^{u}\otimes\id^{(u,1)}\otimes\id^{(u,2)} A_{(u, (u,1),
(u,2))}^{*}\right)
$$
Since $h_\alpha =\alpha\id$ is solution of (\ref{dysys}), one gets
$$ \Tr_{\{u\}}\left(\prod_{u\in \Lambda_n}A_{(u, (u,1), (u,2))} \id^{u}\otimes\id^{(u,1)}\otimes\id^{(u,2)} A_{(u, (u,1), (u,2))}^{*}\right)=
\alpha^{-1}\id^{(u)}.
$$
Then
\begin{eqnarray*}
\Tr_{\Lambda_{[0,n]}}\left(K_{[n,n+1]} a K_{[n,n+1]}^{*}\right)& =&\alpha^{-(|\Lambda_n|-1)} \Tr_{\{x\}}\left(A_{(x, (x,1), (x,2))} a_{x}\otimes a_{(x,1)}\otimes
a_{(x,2)} A_{(x, (x,1), (x,2))}^{*}\right)\\
& =& \alpha^{-(|\Lambda_n|)} \Tr_{\{x\}}\left(\alpha^{1/2}A_{(x, (x,1), (x,2))} a_{x}\otimes a_{(x,1)}\otimes a_{(x,2)} A_{(x, (x,1),
(x,2))}^{*}\alpha^{1/2}\right)\\
&=&\alpha^{-2^{n}}\mathcal{E}^{\alpha}_{(x,(x,1), (x,2))}(a_x\otimes a_{(x,1)}\otimes a_{(x,2)}).
\end{eqnarray*}
Therefore,
\begin{eqnarray*}
 {\varphi_{\alpha}}(a)&=& \alpha^{2^{n+1}-1}\alpha^{-2^{n}}\Tr\left(\prod_{i=0}^{n-1}K_{[i,i+1]} a_{x_0}\otimes \cdots\otimes a_{x_{n-1}}
 \otimes \mathcal{E}^{\alpha}_{(x,(x,1), (x,2))}(a_x\otimes a_{(x,1)}\otimes a_{(x,2)})\prod_{i=0}^{n-1}K_{[i,i+1]}^{*}\right)\\
&=& \alpha^{2^{n}-1}\Tr\left(\prod_{i=0}^{n-1}K_{[i,i+1]} a_{x_0}\otimes \cdots\otimes a_{x_{n-1}}
 \otimes \mathcal{E}^{\alpha}_{(x,(x,1), (x,2))}(a_x\otimes a_{(x,1)}\otimes a_{(x,2)})\otimes\prod_{i=0}^{n-1}K_{[i,i+1]}^{*}\right)\\
&=& {\varphi_{\alpha}}_{\lceil\mathcal{A}_{P(x)\cup \{x\}}}\circ E^{\alpha}_x (a).
\end{eqnarray*}
This completes the proof.
\end{proof}

\begin{remark} The transition expectation $\mathcal{E}^\alpha_x$ acts on the algebra  $\mathcal{A}_{(x, (x,1), (x,2))}$. Its conditional density matrix
$\alpha|A_{(x,(x,1)(x_2)}|^2$ involves both nearest neighbors interactions and competing interactions then it realises a cycle $(x\sim(x,1)\sim(x,2)\sim
x)$. It follows that, the quantum Markov state $\varphi_\alpha$  does not have a one-dimensional representation.
\end{remark}

We point out that in \cite{MBS161} a phase transitions for quantum Markov chains associated with the considered
model has been explored in details.

\section{Structure of Quantum Markov states on trees}\label{sec_stru_QMS}

   Let  $\mathcal{E}_{\{x\}\cup {S}(x)}$ be a transition expectation from $\mathcal{A}_{\{x\}\cup S(x)}$ into $ \mathcal{A}_{x}$.
It follows that the map
\begin{equation}\label{localiz}
\mathcal{E}_{[n,n+1]}:= \bigotimes_{x\in \Lambda_n}\mathcal{E}_{\{x\}\cup {S}(x)}
\end{equation}
is a transition expectation from $\mathcal{A}_{\Lambda_{[n,n+1]}}$ into $\mathcal{A}_{\Lambda_n}$.
 Define
\begin{equation}\label{E_x_ExS(x)}
  E_x = id_{\mathcal{A}_{P(x)}}\otimes \mathcal{E}_{\{x\}\cup {S}(x)}.
\end{equation}
One can see that, for each $x,y\in \Lambda_n$   the maps $E_x$ and $E_y$ commute. Then
\begin{equation}\label{E_n}
  E_n : = \prod_{x\in \Lambda_n}E_x
\end{equation}
is well defined and it satisfies
$$
E_n = id_{\Lambda_{n-1]}}\otimes \mathcal{E}_{\Lambda_{[n,n+1]}}.
$$
\begin{lemma}\label{lem_ExE_n}
  Let $\varphi$ be a localized QMS on $\mathcal{A}_{\Lambda_{[0,n+1]}}$. For the above notations, the following assertions are equivalents
\begin{enumerate}
\item[(i)]  $\varphi\circ E_n = \varphi$.
\item[(ii)]For each $x\in \Lambda_n$ the restriction $\varphi_{\lceil\mathcal{A}_{P(x)\cup \{x\}\cup {S}(x)}}$ satisfies
$$\varphi_{\lceil\mathcal{A}_{P(x)\cup \{x\}\cup {S}(x)}} =\varphi_{\lceil\mathcal{A}_{\{x\}\cup {S}(x)}}\circ E_x.$$
\end{enumerate}
\end{lemma}

\begin{proof}
  $(i) \Rightarrow (ii)$ is straightforward.
$(ii) \Rightarrow (i)$  Let $x\in \Lambda_n$, one has
$$
\varphi\circ E_{n}= \varphi\circ E_x \circ \prod_{y\in \Lambda_n\setminus \{x\}}E_y = \varphi\circ \prod_{y\in \Lambda_n\setminus \{x\}}E_y.
$$
Iterating this procedure one gets (i).
\end{proof}

\begin{theorem}\label{struc-MS}
Any localized QMS $\varphi$ on $\mathcal{A}_V$ defines a pair
$\{\varphi_0 \ , \ (\mathcal{E}_{\Lambda_{[n+1,n]}})\}$ with the following properties:
\begin{enumerate}
\item[(i)]  $\varphi_0$ is a state on $\mathcal{A}_{o}$ and, for all $\,n\in\mathbb N$,
$\mathcal{E}_{\Lambda_{[n, n+1]}}:\mathcal{A}_{\Lambda_{[n, n+1]}}\to \mathcal{A}_{\Lambda_{n}} $ is a localized
 Markov transition expectation;
\item[(ii)] For every $n\in\mathbb N$, the restriction of $\varphi_{\lceil \mathcal{A}_{[0,n]}}=: \varphi_{\Lambda_{[0,n]}}$ on
$\mathcal{A}_{[0,n+1]}$  is characterized by the property:
$$
\varphi_{\Lambda_{[0,n+1]}}(a_0\otimes a_1\otimes \cdots\otimes a_n)
$$
\begin{equation}\label{5.2.1b}
:=\varphi_0(\mathcal{E}_{\Lambda_{[1,0]}}(a_0\mathcal{E}_{\Lambda_{[2,1]}}(a_1(\cdots \mathcal{E}_{\Lambda_{[n,n-1]}}(a_{n-1}\mathcal{E}_{\Lambda_{[n, n+1]}}(a_n)))
\end{equation}
for any $a_i \in \mathcal{A}_{\Lambda_i}$, $0\leq i\leq n$, is a state such that
$$
\varphi_{[0,0]}:=\varphi_{0}
$$
for any $n\in\mathbb N$ and any $a_i \in \mathcal{A}_{\Lambda_i}$, $0\leq i\leq n$.
\end{enumerate}
Conversely, given a pair $\{\varphi_0 \ , \ (\mathcal{E}_{\Lambda_{[n, n+1]}})\}$ satisfying
conditions (i),(ii) above, there exists a unique localized QMS $\varphi$ on $\mathcal{A}_V$
whose associated pair, according to the first part of the theorem, is
$\{\varphi_0 \ , \ (\mathcal{E}_{\Lambda_{[n, n+1]}})\}$.
\end{theorem}

\begin{proof} Necessity.
Let $\{ E_x\}_x$ be a family of quasi-conditional expectation w.r.t. the triplet $\mathcal{A}_{P(x)\cup \{x\}\cup {S}(x)}\supset \mathcal{A}_{P(x)\cup \{x\}}\supset   \mathcal{A}_{P(x) }$ associated with the quantum Markov state $\varphi$.

From Lemma \ref{lem_ExE_n}, the map $E_{[0,n]}:= \prod_{x\in W_n}E_x$ defines a quasi-conditional expectation w.r.t. the triplet $\mathcal{A}_{\Lambda_{n-1]}}\subset \mathcal{A}_{\Lambda_{[0,n]}}\subset \mathcal{A}_{\Lambda_{[0,n+1]}}$ satisfying $$
\varphi_{\lceil\mathcal{A}_{\Lambda_{[0,n]}}}\circ E_{\Lambda_{[0,n]}} = \varphi_{\lceil\mathcal{A}_{\Lambda_{[0,n+1]}}}.
$$
 Since the map $E_{\Lambda_{[0,n]}}$ acts trivially on the algebra $\mathcal{A}_{\Lambda_{n-1]}}$ then it can be written in the form
$$
E_{\Lambda_{[0,n]}} = id_{\mathcal{A}_{\Lambda_{n-1]}}}\otimes\mathcal{E}_{\Lambda_{[n, n+1]}}
$$
where $\mathcal{E}_{\Lambda_{[n, n+1]}} = {E_{[0,n]}}_{\lceil\mathcal{A}_{\Lambda_{[n,n+1]} }}$.
Similarly, for each $x\in \Lambda_n$ the quasi-conditional expectation $E_x$ has the form $E_x = id_{P(x)}\otimes \mathcal{E}_{\{x\}\cup {S}(x)}$ with $\mathcal{E}_{\{x\}\cup {S}(x)} $ is a transition expectation from $\mathcal{A}_{\{x\}\cup {S}(x)}$ into $\mathcal{A}_{\{x\}}$. Due to the tree structure $\Lambda_{n-1]}= \bigsqcup_{x\in \Lambda_n}P(x)$ and $\Lambda_{[n,n+1]} = \bigsqcup_{x\in\Lambda_n}\{x\}\cup S(x)$ where $\bigsqcup$ means disjoint union. It follows that
$$
E_{\Lambda_{[0,n]}} = \prod_{x\in\Lambda_n}E_x = \id_{\mathcal{A}_{\Lambda_{n-1]}}}\otimes\bigotimes_{x\in \Lambda_n}\mathcal{E}_{\{x\}\cup {S}(x)}.
$$
Therefore   $\mathcal{E}_{\Lambda_{[n, n+1]}}=\bigotimes_{x\in\Lambda_n}\mathcal{E}_{\{x\}\cup {S}(x)}.$ This proves (i).

Let $a = a_0\otimes a_1\otimes\cdots \otimes a_n,  a_i\in \mathcal{A}_{\Lambda_i}$.
\begin{eqnarray*}
\varphi(a)  &=& \varphi\circ E_{[0,n]}(a\otimes \id_{\Lambda_{n+1}})\\
&=& \varphi\circ  E_{n-1]} \circ E_{[0,n]}(a\otimes \id_{\Lambda_{n+1}})\\
&\vdots&\\
&=& \varphi_{0}\circ E_{\Lambda_{0]}}\circ  E_{\Lambda_{1]}}\circ \cdots \circ E_{\Lambda_{n-1]}}\circ E_{\Lambda_{[0,n]}}(a\otimes\id_{\Lambda_{n+1}})\\
&=& \varphi_0(E_{\Lambda_{0]}}(a_{\Lambda_0}\otimes  E_{\Lambda_{1]}}(a_{\Lambda_1} \cdots  E_{\Lambda_{n-1]}}(a_{\Lambda_{n-1}}\otimes E_{\Lambda_{[0,n]}}(a_{\Lambda_{n}}\otimes\id_{\Lambda_{n+1}})))))\\
&=& \varphi_0(\mathcal{E}_{\Lambda_{[0,1]}}(a_{\Lambda_0}\otimes  \mathcal{E}_{\Lambda_{[1,2]}}(a_{\Lambda_1} \cdots  \mathcal{E}_{\Lambda_{[n,n+1]}}(a_{\Lambda_{n-1}}\otimes \mathcal{E}_{\Lambda_{[n,n+1]}}(a_{\Lambda_{n}}\otimes\id_{\Lambda_{n+1}}))))).
\end{eqnarray*}
This proves (ii).

Sufficiency. Let $\{\varphi_0, \mathcal{E}_{\Lambda_{[n, n+1]}} \}$ be a pair satisfying (i) and (ii). By (i), the transition expectation
$\mathcal{E}_{\Lambda_{[n, n+1]}} $ is  localized. There exist transition expectations $\mathcal{E}_{\{x\}\cup {S}(x)}$ such that $\mathcal{E}_{\Lambda_{[n, n+1]}}  = \bigotimes_{x\in \Lambda_n}\mathcal{E}_{\{x\}\cup {S}(x)}$. The map $E_x:= \id_{P(x)}\otimes \mathcal{E}_{\{x\}\cup {S}(x)}$ is a quasi-conditional expectation with respect to the triplet $\mathcal{A}_{P(x)\cup \{x\}\cup {S}(x)}\supset \mathcal{A}_{P(x)\cup \{x\}}\supset   \mathcal{A}_{P(x) }$.
The right hand side of  (\ref{5.2.1b}) defines a unique state $\varphi_{\Lambda_{[0,n]}}$ on the algebra $\mathcal{A}_{\Lambda_{[0,n]}}$
 Let $a_x\in \mathcal{A}_{\{x\}\cup {S}(x)}$
for $y\in \Lambda_{n}\setminus \{x\}$ one has $E_y(a_x) =a_x$.

 Then
$$
{\varphi_{\Lambda_{[0,n]}}}_{\lceil\mathcal{A}_{\{x\}\cup P(x)}}\circ E_x(a_x) =   \varphi_{\Lambda_{[0,n]}}\circ E_{\Lambda_{[0,n]}}(a_x\otimes\id) = {\varphi_{\Lambda_{[0,n]}}}(a_x\otimes \id) ={\varphi_{\Lambda_{[0,n]}}}_{\lceil \mathcal{A}_{\{x\}\cup  {S}(x)}}(a_x).
$$
Let $a\in \mathcal{A}_{\Lambda_{[0,n]}}$, since the maps $\mathcal{E}_{\Lambda_{[n,n+1]}}$ are identity preserving then for $k> n$
\begin{eqnarray*}
\varphi_{\Lambda_{k]}}(a\otimes \id )& = & \varphi_0(\mathcal{E}_{\Lambda_{[1,0]}}(a_0\mathcal{E}_{\Lambda_{[2,1]}}(a_1(\cdots \mathcal{E}_{\Lambda_{[n,n-1]}}(a_{n-1}\mathcal{E}_{\Lambda_{[n, n+1]}}(a_n\otimes \cdots \mathcal{E}_{\Lambda_{[k, k+1]}}(\id)))))))\\
&=&\varphi_0(\mathcal{E}_{\Lambda_{[1,0]}}(a_0\mathcal{E}_{\Lambda_{[2,1]}}(a_1(\cdots \mathcal{E}_{\Lambda_{[n,n-1]}}(a_{n-1}\mathcal{E}_{\Lambda_{[n, n+1]}}(a_n\otimes \cdots \mathcal{E}_{\Lambda_{[k-1, k]}}(\id)))))))\\
& =&\varphi_0(\mathcal{E}_{\Lambda_{[1,0]}}(a_0\mathcal{E}_{\Lambda_{[2,1]}}(a_1(\cdots \mathcal{E}_{\Lambda_{[n,n-1]}}(a_{n-1}\mathcal{E}_{\Lambda_{[n, n+1]}}(a_n))))))\\
&=& \varphi_{\Lambda_{k]}}(a).
\end{eqnarray*}

It follows that, the limit  $\varphi:= \lim_{k}\varphi_{\Lambda_{k]}}$  exists, i.e.
$$
\varphi(a) = \varphi_{\Lambda_{[0,n]}}(a).
$$
The functional $\varphi$ is then a state on the algebra $\mathcal{A}_V$ satisfying
$$
\varphi_{\lceil \mathcal{A}_{\{x\}\cup {S}(x)}} \circ E_x = {\varphi_{\Lambda_{[0,n]}}}_{\lceil \mathcal{A}_{\{x\}\cup {S}(x)}} \circ E_x = \varphi_{\lceil \mathcal{A}_{\{x\}\cup {S}(x)}}.
$$
Therefore, the state  $\varphi$ satisfies (\ref{Markov_state_eq}). This completes the proof.
\end{proof}

Finally, we point out that using the argument of a main result of \cite{MS20}, we can prove the following result.

\begin{theorem}\label{main_thm}
Let $\varphi\in \mathcal S(\mathcal{A}_V)$ be a localized quantum Markov
state. Then there exists a diagonal algebra $\mathcal
D_V\subset\mathcal{A}_V$, a Markov random field $\mu$ on
$\mathrm{spec}(\mathcal D_V)$ and a Umegaki conditional
expectation $\mathfrak E: \mathcal{A_V} \to \mathcal D_V$ such
that
\begin{equation}\label{diag_eq}
\varphi = \varphi_\mu\circ \mathfrak E
\end{equation} where $\varphi_\mu$
is the state on $\mathcal D_V$ corresponding to $\mu$.
\end{theorem}

We notice that the diagonalizability result for translation invariant quantum
Markov states first appeared in \cite{GZ} for homogeneous processes
on the forward chains. In \cite{FM} the
proof of diagonalizability has been proved for one-dimensional non-homogeneous QMS. Our result will allow to investigate
QMS over networks which will be a topic of our coming investigations.

\section{(Sub) Quantum Markov states on trees}\label{sec_SubMS}

Let us, before start this section, notice that any restriction of QMS in the sense of Definition \ref{QMS1}, is not necessarily a QMS. It turns out that localized QMS has the mentioned property. To establish such a property, let us recall ceratin auxiliary definitions.

Let $T^{'} = (V^{'}, E^{'})$ be a subtree of the tree. There exists a unique vertex $o'\in V'$ such that $d(o, V') = d(o,o')$. This vertex $o'$ will be referred as a root of the subtree $T^{'}$.
\begin{definition}
  Let $\varphi$ be a localized QMS on $\mathcal{A}_V$. The restriction of $\varphi$ on the algebra $\mathcal{A}_{V'}$ is called Sub-QMS  associated with the subtree $T'$.
\end{definition}

\begin{theorem}\label{aq}
  Any sub-QMS is itself a localized QMS.
\end{theorem}

\begin{proof}
 Let  $T^{'} = (V^{'}, E^{'})$ be a subtree with root $o^{'}$. Let $\varphi$ be a Markov state on the algebra $\mathcal{A}_V$  associated with a family $\{E_x\}_{x\in V}$ of quasi-contional expectation w.r.t. the triplet $\mathcal{A}_{P(x) }\subset \mathcal{A}_{P(x)\cup \{x\}}\subset \mathcal{A}_{P(x)\cup \{x\}\cup {S}(x)}$. Let $x\in V^{'}$. The vertex $o^{'}$ belongs to the unique edge path joining the root $o$ and the vertex $x$. Then the set $P'(x) = P(x)\cap V^{'}$ consists of the elements of the edge path joining $o'$ and $x$. Therefore, the restriction $E^{'}_x$ of $E_x$ on the algebra $\mathcal{A}_{P^{'}(x) \cup \{x\}\cup{S}(x)}$ is a quasi-conditional expectation with respect to the triplet  $\mathcal{A}_{P^{'}(x) }\subset \mathcal{A}_{P^{'}(x)\cup \{x\}}\subset \mathcal{A}_{P^{'}(x)\cup \{x\}\cup {S}(x)}$ and it satisfies
 \begin{equation}\label{phi'MP}
 \varphi^{'}_{\lceil \mathcal{A}_{P^{'}(x)\cup \{x\}}}\circ E^{'}_x =  \varphi^{'}_{\lceil \mathcal{A}_{P^{'}(x)\cup \{x\}\cup {S}(x)}}.
 \end{equation}
 It follows that the state $\varphi^{'}$ is a QMS with the family $\{E^{'}_x\}_{x\in V^{'}}$.
\end{proof}

In the sequel, we reduce ourselves to  the case of regular trees (the Cayley trees). The Cayley tree of order $k$ is characterized by being a tree for which every vertex has exactly $k+1$ nearest-neighbors. We consider the semi-infinite Cayley tree $\Gamma^k_+= (V,E)$ with root $o$. In this case, any vertex has exactly $k$ direct successors denoted $(x,i), i=1,2,\cdots, k.$
$$
\overrightarrow{S}(x) = \{(x,1), (x,2), \cdots, (x,k)\}
$$
It follows that the coordinate structure on the tree gives
$$
\Lambda_n = \{ (i_1, i_2, \cdots, i_n); \quad i_j = 1,2, \cdots, k \}.
$$
The coordinate structure of the semi-infinite Cayley tree allows to introduce a shift on it (see \cite{fannes}). Namely, for $x = (i_1, i_2, \cdots, i_n) \in \Lambda_n $, we define $k$ shifts on the tree as follows
\begin{equation}\label{shifts}
  \alpha_j(x) = (j, x)= (j, i_1, i_2, \cdots, i_n)\in \Lambda_{n+1}.
\end{equation}

$$
\alpha_x := \alpha_{i_1}\circ \alpha_{i_2}\circ\cdots\circ \alpha_{i_n}.
$$
The shift $\alpha_x$ maps the semi-infinite Cayley tree $\Gamma^k_+$ onto  its subtree $T_x$ defined by (\ref{S_1n_S_infty}).

\begin{remark}
We point out that on  the (regular) Cayley tree, one can define translations via free group structure (see \cite[Chapter 1]{Roz} for details). However, in our setting, this can not be applied, since the considered semi-infinite Cayley has a root.
\end{remark}

 One has $\alpha_x(o) = x$ and $\alpha_x(\Lambda_n)= S_n(x)$.
The shifts $\alpha_j$ can be extended to  the algebra $\mathcal{A}_V$ as follows:
\begin{equation}\label{shift_algebra}
\alpha_j\left( \bigotimes_{ x\in \Lambda_{[0,n]}}a_x \right):=
 \id^{(o)}\otimes\bigotimes_{ x\in \Lambda_{[0,n]}} a_x^{(j,x)}.
\end{equation}
\begin{definition}
  A state $\varphi$ on $\mathcal{A}_{L}$ is said to be translation invariant if
  \begin{equation}\label{trans_inv_state}
    \varphi\circ \alpha_j = \varphi
  \end{equation}
 for every $j \in \{1,2,\cdots, k\}$.
\end{definition}

\begin{theorem}\label{thm_tinv_qmc} Let $\varphi$ be a localized QMS associated with a family $\{\mathcal{E}_{\{x\}\cup {S}(x)}$ of transition expectations. The following assertions are equivalent.
\begin{description}
  \item[(i)] The state $\varphi$ is translation invariant.
  \item[(ii)] The sub Markov state $\varphi_{T_x}$ on the subtree with vertex set $T_x$  given by (\ref{S_1n_S_infty}) satisfies
  \begin{equation}\label{phitxa=phia}
  \varphi_{T_x}(\alpha_x(a)) = \varphi(a)
  \end{equation}
  for all $a\in \mathcal{A}_V$.
 \item[(iii)]  There exists a completely positive identity preserving map $\mathcal{E}: M^{\otimes (k+1)} \to M$ such that  the transitions expectations $\mathcal{E}_{\{x\}\cup {S}(x)}$ are copies of $\mathcal{E}$ \\
      i.e.
      \begin{equation}\label{E=alpha(E0)}
      \mathcal{E}_{\{x\}\cup {S}(x)}\circ\alpha_x(a) = \alpha_x\circ \mathcal{E}_{\{o\}\cup {S}(o)}(a)
      \end{equation}
      for all $a\in \mathcal{A}_{\Lambda_{1]}}$.
\end{description}
\end{theorem}

\begin{proof}
(i)$\Leftrightarrow$ (ii) Let $x\in L$, according to the above defined coordinate structure $x = (i_1, i_2, \cdots, i_n)$ where $n = d(x,o)$ and $i_1, i_2, \cdots, i_n\in \{1,2,\cdots, i_k\}$. Since $\varphi_{T_x}$ is the restriction of $\varphi$ on $\mathcal{A}_{T_x}$ then  (\ref{trans_inv_state}) leads to
\begin{eqnarray*}
\varphi_{T_x}(\alpha_x(a)) &=& \varphi\circ \alpha_{i_1}\circ\alpha_{i_2}\circ\cdots \alpha_{i_{n-1}}\circ \alpha_{i_n}(a)\\
&=& \varphi\circ\alpha_{i_2}\circ\cdots \alpha_{i_{n-1}}\circ \alpha_{i_n}(a)\\
&\vdots&\\
&=& \varphi\circ \alpha_{i_n}(a)\\
&=& \varphi(a).
\end{eqnarray*}
Conversely, by applying (\ref{phitxa=phia}) on elements of $\Lambda_1$ one gets (\ref{trans_inv_state}).\\
(ii)$\Leftrightarrow$ (iii)
 By (\ref{phitxa=phia}) for each $b\in \mathcal{A}_o$  one has
$$
\varphi_x(\alpha_x(b)) = \varphi_o(b)
$$
and
$$
\varphi_o \circ\mathcal{E}_{\{o\}\cup {S}(o)}(a) = \varphi_{T_x}(\alpha_x(\mathcal{E}_{\{o\}\cup {S}(o)}(a)))
= \varphi_{x}(\alpha_x(\mathcal{E}_{\{o\}\cup {S}(o)}(a)))
$$
The map
\begin{equation}
\tilde{\mathcal{E}}_{\{x\}\cup{S}(x)}(a_{\{x\}\cup{S}(x)})
 := \alpha_x(\mathcal{E}_{\{o\}\cup {S}(o)}(\alpha_x^{-1}( a_{\{x\}\cup{S}(x)})))
\end{equation}
realizes a quasi-conditional expectation from $\mathcal{A}_{\{x\}\cup {S}(x)}$ into $\mathcal{A}_x$.
 And it satisfies
 \begin{eqnarray*}
 \varphi_{x}\circ \tilde{\mathcal{E}}_{\{x\}\cup{S}(x)}(a_{\{x\}\cup{S}(x)})
 &=& \varphi_{T_x}\circ\alpha_x(\mathcal{E}_{\{o\}\cup {S}(o)}(\alpha_x^{-1}( a_{\{x\}\cup{S}(x)})))\\
&=& \varphi\circ \mathcal{E}_{\{o\}\cup {S}(o)}(\alpha_x^{-1}( a_{\{x\}\cup{S}(x)}))\\
&=& \varphi(\alpha_x^{-1}( a_{\{x\}\cup{S}(x)}))\\
&=& \varphi_{T_x}(\alpha_x(\alpha_x^{-1}( a_{\{x\}\cup{S}(x)})))\\
&=& \varphi_{T_x}( a_{\{x\}\cup{S}(x)}).
 \end{eqnarray*}
Therefore, the pair $\{\varphi_o, \tilde{\mathcal{E}}_{\{x\}\cup{S}(x)})\}$ is associated with the Markov state $\varphi$ in addition to the
pair $\{\varphi_o, \mathcal{E}_{\{x\}\cup S(x)}\}$. From Theorem \ref{5.2.1b} we conclude that the two pair coincide. This means that
 $$\mathcal{E}_{\{x\}\cup S(x)}= \alpha_x\circ\mathcal{E}_{\{o\}\cup {S}(o)}\circ\alpha_x^{-1}$$
This leads to (iii).\\
Conversely, let $a=a_{\Lambda_o}\otimes a_{\Lambda_1}\otimes\cdots a_{\Lambda_n}\in\mathcal{A}_{\Lambda_{[0,n]}}$.
 Since $\alpha_x(\Lambda_j) = S_{j}(x)$, by (\ref{E=alpha(E0)}) one has
\begin{eqnarray*}
\varphi_{T_x}(\alpha_x(a))&=& \varphi_{T_x}\left(\mathcal{E}_{S_{[0,1]}(x)}\left(\alpha_x(a_{\Lambda_o})\otimes
 \mathcal{E}_{S_{[1, 2]}(x)}\left(\alpha_x(a_{\Lambda_1})\cdots \mathcal{E}_{S_{[n-1,1n]}(x)}\left(\alpha_x(a_{\Lambda_n})
 \otimes\id\right)\right)\right)\right)\\
 &=&\varphi_x\left( \alpha_{x}(\mathcal{E}_{\Lambda_{[0,1]}})\left(\alpha_x(a_{\Lambda_0})\otimes\alpha_x(\mathcal{E}_{\Lambda_{[1,2]}})\left(
 \alpha_x(a_{\Lambda_1})\cdots \alpha_x(\mathcal{E}_{ \Lambda_{[n-1,n]}}) \left(\alpha_x(a_{\Lambda_n})\otimes\id\right)\right)\right)\right)\\
 &=&\varphi_x\circ\alpha_{x}\left( \mathcal{E}_{\Lambda_{[0,1]}} \left(a_{\Lambda_0}\otimes\mathcal{E}_{\Lambda_{[1,2]}}\left(
 a_{\Lambda_1}\cdots\mathcal{E}_{ \Lambda_{[n-1,n]}}  \left(a_{\Lambda_n}\otimes\id\right)\right)\right)\right)\\
 &=& \alpha_{0}\left( \mathcal{E}_{\Lambda_{[0,1]}} \left(a_{\Lambda_0}\otimes\mathcal{E}_{\Lambda_{[1,2]}}\left(
 a_{\Lambda_1}\cdots\mathcal{E}_{ \Lambda_{[n-1,n]}}  \left(a_{\Lambda_n}\otimes\id\right)\right)\right)\right)\\
 &=&\varphi(a).
 \end{eqnarray*}
 This finishes the proof.
 \end{proof}

 \begin{remark}
   According to  Theorem \ref{thm_tinv_qmc}, a translation invariant localized quantum Markov state on the Cayley tree is characterized by a pair $(\varphi_o,  \mathcal{E})$ of initial state on the algebra $\mathcal{A}_o$ and a transition expectation $\mathcal{E}$ from $ M^{\otimes n+1} $ into $M$.
 \end{remark}

\section{Acknowledgments}{The authors gratefully acknowledge Qassim University, represented by the Deanship of Scientific Research,
 on the financial support for this research under the number (10173-cba-2020-1-3-I) during the academic year 1442 AH / 2020 AD.}

\end{document}